\documentclass[11pt,bezier]{article}
\usepackage{amsmath,amssymb,amsfonts,euscript}

\usepackage{graphicx}

\usepackage[ansinew]{inputenc}
\usepackage{graphicx}
\usepackage{color}
\usepackage[colorlinks]{hyperref}
\usepackage[active,new,noold,marker]{xrcs}
\catcode`\=13
\def{$\bowtie$}

\usepackage{eurosym}
\usepackage{lscape} 

\newtheorem{prethm}{{\bf Theorem}}

\newenvironment{thm}{\begin{prethm}{\hspace{-0.5
               em}{\bf}}}{\end{prethm}}

\newtheorem{prepro}{{\bf Theorem}}

\newtheorem{preque}{{\bf Question}}

\newtheorem{preprop}{{\bf Proposition}}

\newtheorem{precor}{{\bf Corollary}}

\newenvironment{cor}{\begin{precor}{\hspace{-0.5
               em}{\bf}}}{\end{precor}}

\newtheorem{preconj}{{\bf Conjecture}}

\newtheorem{predefi}{{\bf Observation}}

\newenvironment{defi}{\begin{predefi}{\hspace{-0.5
               em}{\bf}}}{\end{predefi}}

\newtheorem{preprob}{{\bf Problem}}

\newenvironment{prob}{\begin{preprob}{\hspace{-0.5
               em}{\bf.}}}{\end{preprob}}

\newtheorem{preremark}{{\bf Remark}}

\newtheorem{preexample}{{\bf Example}}

\newtheorem{prelem}{{\bf Lemma}}

\newenvironment{lem}{\begin{prelem}{\hspace{-0.5
               em}{\bf}}}{\end{prelem}}

\newtheorem{prelam}{{\bf Lemma}}

\newtheorem{preproof}{{\bf Proof}}

\newenvironment{proof}[1]{\begin{preproof}{\rm
               #1}\hfill{$\Box$}}{\end{preproof}}

\newtheorem{preali}{{\bf Proof of Theorem 1.}}

\newenvironment{ali}[1]{\begin{preali}{\rm
               #1}\hfill{$\Box$}}{\end{preali}}

\newtheorem{prealy}{{\bf Proof of Theorem 2.}}

\newenvironment{aly}[1]{\begin{prealy}{\rm
               #1}\hfill{$\Box$}}{\end{prealy}}

\title{ The Complexity of the Proper Orientation Number}

\author{{\normalsize
{\sc A. Ahadi${}^{\mathsf{a}}$},\,
 {\sc A. Dehghan${}^{\mathsf{b}}$},\,
}\vspace{3mm}
\\{\footnotesize{${}^{\mathsf{a}}$\it Department of
Mathematical Sciences, Sharif University of Technology, Tehran,
Iran}}  {\footnotesize{}}\\{\footnotesize{${}^{\mathsf{b}}$\it
Department of Mathematics and Computer Science,
Amirkabir University of Technology, Tehran, Iran}}
\thanks{{\it E-mail addresses}:  $\mathsf{arash\_ahadi@mehr.sharif.edu}$,  $\mathsf{ali\_dehghan16@aut.ac.ir}$. } }

\date{
}

\begin{document}
\maketitle

\begin{abstract}
{\small \noindent
Graph orientation is a well-studied area of graph theory. A proper orientation of a graph $G = (V,E)$ is an orientation $D$ of $E(G)$ such that
for every two adjacent vertices $ v $ and $ u $, $  d^{-}_{D}(v) \neq d^{-}_{D}(u)$ where $d_{D}^{-}(v)$ is the
number of edges with head $v$ in $D$. The proper orientation number of $G$ is defined as $  \overrightarrow{\chi} (G) =\displaystyle \min_{D\in \Gamma} \displaystyle\max_{v\in V(G)}  d^{-}_{D}(v)  $ where $\Gamma$ is the set of proper orientations of $G$. We have $ \chi(G)-1 \leq  \overrightarrow{\chi} (G)\leq \Delta(G) $.
We show that, it is $ \mathbf{NP} $-complete to decide whether $\overrightarrow{\chi}(G)=2$, for a given planar graph $G$. Also, we prove that there is a polynomial time
algorithm for determining the proper orientation number of  $3$-regular graphs. In sharp contrast, we will prove that this problem is $ \mathbf{NP} $-hard for $4$-regular graphs.
}

\begin{flushleft}

\end{flushleft}

\end{abstract}


\section{Introduction}
\label{}

\label{}
Graph orientation is a well-studied area of graph theory,  that provides a connection between directed and undirected graphs \cite{MR2178357}. There are several problems concerned with orienting the edges of an undirected
graph in order to minimize some measures in the resulting
directed graph, for instance see \cite{ MR0491326, MR2876334}. On the other hand, there are many ways to color the vertices of graphs properly. A proper vertex coloring of a digraph $D$ is defined, simply a vertex coloring of its
underlying graph $G$. The chromatic number of a digraph provides interesting information about
its subdigraphs. For instance, a theorem of Gallai   proves
that digraphs with high chromatic number always have long directed paths \cite{MR0233733}.

Venkateswaran \cite{MR2087903} initiated the study of the problem of orienting the edges of a given simple graph so that the maximum indegree of vertices is minimized. Afterwards, Asahiro et al. in \cite{MR2311131} generalized this problem for weighted graphs. It was proved that, this problem can be solved in polynomial-time if all the edge weights are identical \cite{MR2311131, MR2296137, MR2087903}, but it is $ \mathbf{NP} $-hard in general \cite{MR2311131}. Furthermore, the problem can be  solved in polynomial-time if the input
graph is a tree, but for planar bipartite graphs it is  $ \mathbf{NP} $-hard \cite{MR2311131}.
For more information about the recent results about this problem see \cite{DAM}.

On the other hand, in 2004 Karo\'nski, \L{}uczak and Thomason
initiated the study of proper labeling \cite{MR2047539}. They introduced an edge-labeling which
is additive vertex-coloring that means for every edge
$uv$, the sum of labels of the edges incident to $u$ is different
from the sum of labels of the edges incident to $v$ \cite{MR2047539}. Also,
it is conjectured that three labels
$\{1,2,3\}$ are sufficient for every connected graph, except
$K_2$ (1, 2, 3-Conjecture, see \cite{MR2047539}). This labeling have been studied extensively by several authors, for instance see \cite{MR2145514,   MR2404230, z2,David, MR2595676}.
Afterwards, Borowiecki et al. consider the directed version of this problem.
Let $D$ be a simple directed graph and suppose that each edge of $D$ is assigned an integer label. For a vertex $v$ of $D$, let $q^+(v)$
 and $q^-(v)$ be the sum of labels lying on the arcs outgoing form $v$ and incoming to $v$, respectively. Let $q(v) =
q^+(v) - q^-(v)$. Borowiecki et al. proved that  there is always a labeling from $\{1,2\}$, such that $q(v)$ is a proper coloring of $D$ \cite{MR2895496}, see also \cite{z3}.

Furthermore, Borowiecki et al. consider another version of above problems and they show that every undirected
graph $G$ can be oriented so that adjacent vertices have different in-degrees \cite{MR2895496}.
In this work, we  consider the problem of orienting the edges of an undirected graph such that adjacent vertices have different in-degrees and the maximum indegree of vertices is minimized in the resulting directed graph.

The
{\it proper orientations} of $G$
which are orientations $D$ of $G$
such that
for every two adjacent vertices $ v $ and $ u $, $  d^{-}_{D}(v) \neq d^{-}_{D}(u)$. The {\it proper orientation number} of $G$ is defined as
$  \overrightarrow{\chi} (G) =\displaystyle \min_{D\in \Gamma} \displaystyle\max_{v\in V(G)}  d^{-}_{D}(v)  $, where $\Gamma$ is the set of proper orientations of $G$.

The proper orientation number is well-defined and every proper orientation of graph introduces a proper vertex coloring for its vertices. Thus, $ \chi(G)-1 \leq  \overrightarrow{\chi} (G)$. On the other hand $  \overrightarrow{\chi} (G)\leq \Delta(G) $. Consequently,

\begin{equation}
\chi(G) -1\leq \overrightarrow{\chi} (G)\leq  \Delta(G).
\end{equation}

In this work,  we focus on regular graphs and planar graphs. We show that there is a polynomial-time
algorithm for determining the proper orientation number of  $3$-regular graphs. But  it is $ \mathbf{NP} $-complete to decide whether the proper orientation number of a given $4$-regular graph is $3$ or $4$.

\begin{thm}
Determining the proper orientation number of a
given $4$-regular graph is $ \mathbf{NP} $-hard; but there is a polynomial-time
algorithm to determine the proper orientation number for  $3$-regular graphs.
\end{thm}

It is easy to see that $ \overrightarrow{\chi}(G)=1 $ if and only if every connected component of $G$ is a star. But for $ \overrightarrow{\chi}(G)=2 $, we have the following:

\begin{thm}
It is $ \mathbf{NP} $-complete to decide $\overrightarrow{\chi}(G)=2$, for a given planar graph $G$.
\end{thm}

In this paper we consider simple graphs
and we refer to \cite{bondy} for standard notation and concepts.
For a graph $G$, we use $n$ and $m$ to denote its numbers of
vertices and edges, respectively. Also, for every $v\in V (G)$, $d(v)$ and $N_G (v)$ denote the degree of $v$ and the neighbor set of $v$, respectively. A {\it spanning subgraph} of a graph $G$ is a subgraph whose vertex set is $V(G)$.
We say that a set of vertices are {\it independent} if there is no edge
 between these vertices.
 The {\it independence number}, $\alpha(G)$ of a graph $G$ is the size of a largest independent set of
 $G$. Also a {\it clique} of a graph is a set of mutually adjacent vertices.
We denote the maximum degree of a graph $G$ by $\Delta(G)$.
A {\it directed graph } $G$ is an ordered pair $(V (G),E(G))$ consisting of a set
$V (G)$ of vertices and a set $E(G)$ of edges, with an incidence function $ D $ that associates with each edge of $G$ an ordered pair of vertices of $G$. If $e=uv$ is an edge and $ G(e) = u \rightarrow v$, then
$e$ is from $u$ to $v$. The vertex $u$ is the tail of $e$, and the vertex $v$ its head.
Note that every orientation $D$ of a graph, introduced a digraph.
The indegree $d_{D}^{-}(v)$ of a vertex $v$ in $D$ is the
number of edges with head $v$ of $v$.
For $ k\in \mathbb{N} $, a {\it proper vertex $k$-coloring} of $G$ is a function $c:
V(G)\longrightarrow \lbrace 1,\ldots , k \rbrace$, such that if $u,v\in V(G)$ are adjacent,
then $c(u)$ and $c(v)$ are different.
The smallest integer $k$ such that
$G$ has a proper vertex $k$-coloring is called the {\it chromatic number} of $G$ and denoted by $\chi(G)$.
.Similarly, for $ k\in \mathbb{N} $, a {\it proper edge $k$-coloring} of $G$ is a function $c:
E(G)\longrightarrow \lbrace 1,\ldots , k \rbrace$, such that if $e,e'\in E(G)$ share a common endpoint,
then $c(e)$ and $c(e')$ are different.
The smallest integer $k$ such that
$G$ has a proper edge $k$-coloring is called the {\it edge chromatic number} of $G$ and denoted by $\chi '(G)$. By Vizing's theorem \cite{MR0180505}, the edge chromatic number of a graph $G$ is equal to either $ \Delta(G) $ or $ \Delta(G) +1 $. Those graphs
$G$ for which $\chi '(G)=\Delta(G)  $ are said to belong to $Class$ $1$, and the others to $Class$ $2$.
For a graph $G=(V,E)$, the {\it line graph} $G$ is denoted by $L(G)$, is a graph with the set of vertices $E(G)$ and two vertices of $L(G)$ are adjacent if and only if their corresponding edges share a common endpoint in $G$.
For simplicity and with a slight abuse of notation, we also denote by $D$ the digraph resulting from an orientation $D$ of the graph $G$.

\section{Regular graphs}

Let $G$ be an $r$-regular graph  and suppose that $D$ is a proper orientation of $G$ with maximum indegree $ \overrightarrow{\chi} (G)  $. We have
$ \overrightarrow{\chi} (G)> \frac{1}{n}  \sum_{v\in V(G)}d^{-}_{D}(v)=\frac{r n/ 2}{n}$.
So we have the following simple observation about regular graphs.

 \begin{defi}
For every $r$-regular graph $G $ with $ r\neq 0 $, $ \overrightarrow{\chi} (G)\geq \lceil \frac{r+1}{2}\rceil$.
\end{defi}

Here, we present a lemma about $(2k+1)$-regular graphs,
then we use it to prove Theorem $1$.

\begin{lem}
For every $(2k+1)$-regular graph $G$, $ k \in \mathbb{N} $,

 ($i$) $  \overrightarrow{\chi} (L(G))=3k $ if and only if $G$ belongs to $Class$ $1$.

($ii$) $  \overrightarrow{\chi} (G)=k+1$ if and only if $ \chi(G)=2 $.

\end{lem}

\begin{proof}{

($i$) First, let $G$ be a $(2k+1)$-regular graph belonging to $Class$ $1$. $G$ has a proper edge coloring $ c $, such that $ c:E(G)\rightarrow \lbrace 1,\ldots , 2k+1\rbrace $. Let $\lbrace e_{1},\ldots , e_{m} \rbrace$ be the set of edges of $G$ and $ V(L(G))=E(G) $. Now, we present a proper orientation for $L(G)$ with maximum indegree $ 3k $. For every edge $ e_i e_j \in E(L(G))$, such that $  c(e_j)- c(e_i)> k $, orient $ e_i e_j $ from $ e_i  $ to $ e_j  $.
Also for every two numbers $ p,q \in  \lbrace 1,\ldots , 2k+1\rbrace $, such that $ \vert p - q\vert \leq k $, let $ H_{p,q} $ be the induced subgraph on the vertices $ c^{-1}(p) \cup c^{-1}(q) $ in $L(G)$. Since $G$ is a regular graph that belonging to $Class$ $1$, $ H_{p,q} $ is a spanning $2$-regular subgraph of $L(G)$. Therefore every component of $ H_{p,q} $ is an even cycle. Orient each of these cycles to obtain a directed cycle. We denote the resulting orientation by $D$. We claim that $D$ is a proper orientation of $L(G)$, with maximum indegree $3k$. Since $c$ is a proper edge coloring of $G$, it is sufficient to prove that for every vertex $ e_i \in V(L(G)) $, we have $ d^{-}_{D}(e_i)= c(e_i)+(k-1) $.
Let $e$ be a vertex of $L(G)$.
First suppose that $ c(e) < k $. For every number $i$, $ i\in \lbrace 1,\ldots , 2k+1\rbrace \setminus \lbrace c(e) \rbrace$, the vertex $e$, has exactly two neighbors $ e^{'}_{i} $ and $ e^{''}_{i} $, with $ c(e^{'}_{i})=c(e^{''}_{i})=i $ in $L(G)$. In $D$ for every $i\neq c(e)$, $ 1 \leq i \leq c(e)+k $, $e$ has exactly one incoming edge from the set $ \lbrace e^{'}_{i} , e^{''}_{i} \rbrace $. Also for every $i$, $ c(e)+k < i \leq 2k+1 $, $e$ doesn't have any incoming edge from the set $\lbrace e^{'}_{i} , e^{''}_{i}  \rbrace$. Therefore $ d^{-}_{D}(e_i)= c(e_i)+k-1  $.
For every edge $e$ with the property $ c(e) \geq k $, by a similar argument, we have $ d^{-}_{D}(e_i)= c(e_i)+k-1  $.
Consequently $  \overrightarrow{\chi} (L(G)) \leq 3k $ (in the following we see that if $ \overrightarrow{\chi} (L(G)) \leq 3k  $, then $ \overrightarrow{\chi} (L(G)) = 3k  $).

On the other hand, suppose that $G$ is a $(2k+1)$-regular graph and $  \overrightarrow{\chi} (L(G)) \leq 3k $. Let $ f $ be a proper orientation of $L(G)$ with the maximum indegree at most $  3k $. Let $u$ be an arbitrary vertex of $G$ and $N_G (u)= \lbrace w_1, \ldots , w_{2k+1} \rbrace$. Since $f$ is a proper orientation of $L(G)$ and $\lbrace uw_1, \ldots , uw_{2k+1} \rbrace  $ is a clique in $L(G)$, so $  d_{f}^{-}(uw_1), \ldots , d_{f}^{-}(uw_{2k+1}) $ are distinct numbers. Consequently

\begin{center}
$\displaystyle\sum _{i=1}^{2k+1}d^{-}_{f}(uw_i) \leq  (3k)+(3k-1)+\cdots +(3k-(2k))=2k(2k+1) $,

$ \displaystyle \sum_{v \in V(G)}\displaystyle\sum _{vw\in E(G)}d^{-}_{f}(vw)  \leq   2k(2k+1)n$.
\end{center}

On the other hand, by counting the number of edges in two ways, we have:

\begin{center}
$ \displaystyle\sum_{v \in V(G)}\displaystyle\sum _{  vw\in E(G)}d^{-}_{f}(vw)=2 \big ( \displaystyle\sum_{e \in V(L(G))} d^{-}_{f}(e) \big ) = 2 \vert E(L(G))\vert= 2k(2k+1)n $.
\end{center}

Thus, for every $v\in V(G)$ we have:

\begin{center}
 $ \displaystyle\sum _{ v w\in E(G)}d^{-}_{f}(wv)= 2k(2k+1)$.
\end{center}

It means that for every $v\in V(G)$, $   \lbrace d^{-}_{f}(wv) \vert v w\in E(G) \rbrace = \lbrace k,k+1, \ldots , 3k \rbrace  $. Therefore $\overrightarrow{\chi} (L(G))= 3k$ and also the function $ d^{-}_{f} $ is a proper edge coloring of $G$ with $ 2k+1 $ colors. It means that $G$ belongs to $Class$ $1$.

($ii$) First suppose that $  \overrightarrow{\chi} (G)=k+1$ through a proper orientation $D$. For every $i$, denote the number of vertices $ v\in V(G) $ with the indegree $i$, by $ n_i $. We have: $ \sum_{i=0}^{k+1}n_i =n$ and $\sum_{i=0}^{k+1}i \cdot n_i = m = \frac{2k+1}{2} n  $. So $ n_{k+1} \geq  m - kn =\frac{n}{2}$. The set of vertices $v$ with $ d_{D}^{-}(v)=k+1 $, forms an independent set. Obviously every regular graph with $ \alpha (G) \geq \frac{n}{2} $, is bipartite. Thus, $ \chi(G)=2 $.

Next, let $G[X,Y]$ be a bipartite $(2k+1)$-regular graph. By Observation $1$, $\overrightarrow{\chi} (G) \geq k+1$. By
K\H onig's theorem \cite{MR1367739}, $G$ has a decomposition into $2k+1 $ perfect matchings. Orient the edges of $ k+1 $ perfect matchings from $X$ to $Y$, and other edges from $Y$ to $X$. It is easy to see that this is a proper orientation with maximum indegree $ k+1 $.
}\end{proof}

\begin{cor}
For every $3$-regular graph $G$, other than $ K_4 $, $  \overrightarrow{\chi} (G)=\chi(G) $.
\end{cor}

\begin{ali}{
Clearly, the problem is in NP. It was shown that it is $ \mathbf{NP} $-hard to determine the edge chromatic number of a
cubic graph \cite{MR635430}. By Lemma $1$, for every cubic graph $G$, $  \overrightarrow{\chi} (L(G))=3 $, if and only if $G$ belongs to $Class$ $1$. So determining the proper orientation number of a $4$-regular graph is $ \mathbf{NP} $-hard.
For the second part of the theorem, let $G$ be a $3$-regular graph, other than $K_4$, by Brooks' Theorem \cite{MR0012236} $ \chi(G) \leq 3 $. There is a polynomial-time algorithm for determining whether a given graph $G$ has a chromatic number at most $2$. So by Corollary $1$, there is a polynomial-time
algorithm for determining the proper orientation number of a
given $3$-regular graph.

}\end{ali}

It was shown that it is $ \mathbf{NP} $-hard to determine the edge chromatic number of an
$r$-regular graph for any $r \geq 3$ \cite{MR689264}. So by Lemma $1$ we have:

\begin{cor}
For any $r \geq 1$, the following problem is $ \mathbf{NP} $-hard: "Given a $4r$-regular graph $G$, determine $\overrightarrow{\chi}(G)$".
\end{cor}

Here, we presents a simple $(2-\frac{2}{r+2})$-approximation algorithm
for finding a proper orientation with the minimum of maximum indegree between the proper orientations of $G$, for an $r$-regular graph $G$. It is sufficient to start with a vertex $v$ of maximum degree and for every edge $e=vu$, simply orient $ e $ from $ u $ to $ v $ and next in $ G'=G \setminus \lbrace v \rbrace $ repeat the above procedure.
By Observation $1$ we have  $ \overrightarrow{\chi} \geq \lceil \frac{r+1}{2}\rceil $, so the previous greedy algorithm is a $\theta$-approximation algorithm, where

\begin{center}
$\theta=  2 - \frac{2}{r+2} \geq  \frac{r}{\lceil \frac{r+1}{2}\rceil }$.
\end{center}

\section{Planar graphs}

Let $\Phi$ be a $3$-SAT formula with clauses $C=\lbrace
c_1, \cdots ,c_k\rbrace $ and variables
$X=\lbrace x_1, \cdots ,x_n\rbrace $. Let $G(\Phi)$ be a graph with the vertices $C \cup X \cup (\neg X)$, where $\neg X = \lbrace \neg x_1, \cdots , \neg x_n\rbrace$, such that for each clause $c_j=y \vee z \vee w $, $c_j$ is adjacent to $y,z$ and $w$, also every $x_i \in X$ is adjacent to $\neg x_i$. $\Phi$ is called planar $3$-SAT(type $2$) formula if $G(\Phi)$ is a planar graph. It was shown that the problem of satisfiability of planar $3$-SAT(type $2$) is $ \mathbf{NP}$-complete \cite{zhu1}. In order to prove Theorem $2$, we reduce the following problem to our problem.

\textbf{Problem}: {\em Satisfiability of planar $3$-SAT(type $2$).}\\
\textsc{Input}: A $3$-SAT(type $2$) formula  $ \Phi $.\\
\textsc{Question}: Is there a truth assignment for $ \Phi $ that satisfies all the clauses?\\

\begin{aly}{

Consider an instance of planar $3$-SAT(type $2$) $\Phi$ with variables
$X=\lbrace x_1, \cdots ,x_n\rbrace $ and clauses $C=\lbrace
c_1, \cdots ,c_k\rbrace $. We transform this into a planar graph $\widehat{G}(\Phi)$
such that $ \overrightarrow{\chi} (\widehat{G} (\Phi))=2 $ if and only if $\Phi$ is satisfiable.
We use two auxiliary graphs $ H(x_i)  $ and
$ T(c_j) $ which are shown in Figure 1.
We construct the planar graph $\widehat{G}(\Phi)$ from $G(\Phi)$,
 for every $x_i \in X$ replace the subgraph on $\neg x_i$ and $ x_i$ by $ H(x_i) $, also
replace every clause $c_j $ of $C$ by $ T(c_j) $. Furthermore,
for every clause $c_j$ with the literals $x$, $y$, $z$ add the edges $xc_j^1$, $yc_j^2$ and $zc_j^3$.
Call the resulting graph $\widehat{G}(\Phi)$.
 Clearly $ \overrightarrow{\chi}( \widehat{G}(\Phi)) \geq 2 $.

\begin{figure}[h]
\begin{center}
\includegraphics[scale=.7]{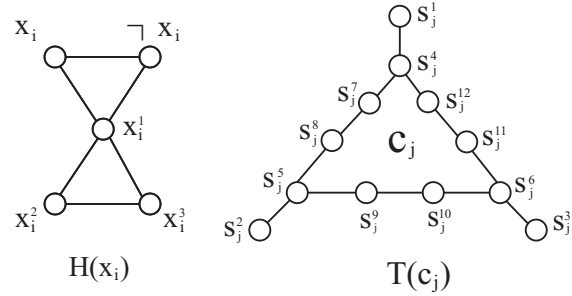}\label{bbbb}
\caption{The two auxiliary graphs $ H(x_i)  $ and $ T(c_j) $}
\end{center}
\end{figure}

First suppose that $ \overrightarrow{\chi} (\widehat{G}(\Phi)) =2 $ and $D$ is a proper orientation of $\widehat{G}(\Phi)$ with maximum indegree $2$. For each variables $x_i$, the subgraph on $x_i^1, x_i^2$ and $x_i^3$ is a triangle, so $\{ d_{D}^{-} (x_i^1),d_{D}^{-} ( x_i^2),d_{D}^{-} (x_i^3)\} =\{0,1,2\}$. Therefore the edges $x_ix_i^1$ and $(\neg x_i)x_i^1$ were oriented from $x_i^1$ to $x_i$ and $\neg x_i$. Consequently $\{ d_{D}^{-} (x_i),d_{D}^{-} (\neg x_i)\} =\{1,2\}$.
Let $x_i$ and $c_j$ be arbitrary variable and clause such that $x_i$ or $\neg x_i $ is used in $c_j$.
Since for every variable $x_i$, $\{ d_{D}^{-} (x_i),d_{D}^{-} (\neg x_i)\} =\{1,2\}$, therefore for every edge $e$ between $ H(x_i) $ and $ T(c_j) $, $e$ was oriented from $ H(x_i) $ to $ T(c_j) $. So we have the following.

\textbf{Fact $1$.} For every edge $e=v_i s^t_j$, $(v_i\in \{x_i, \neg x_i \}, s^t_j \in \{ s^1_j,s^2_j,s^3_j\})$, we have $d_{D}^{-} (s^t_j)= 3- d_{D}^{-} (v_i)$. $\lozenge$

Therefore, for every $t$, $1 \leq t \leq 3 $, two cases can be considered: $d_{D}^{-} (s^t_j)=1$ or $d_{D}^{-} (s^t_j)=2$.\\
\textbullet $ $ If $d_{D}^{-} (s^t_j)=1$, then $s^t_j  s^{t+3}_j   $ was oriented form $s^t_j$ to $s^{t+3}_j$ and so $d_{D}^{-} (s^{t+3}_j)=2$.\\
\textbullet $ $ If $d_{D}^{-} (s^t_j)=2$, then $s^t_j  s^{t+3}_j $ was oriented form $s^{t+3}_j$ to $s^t_j$ and so $d_{D}^{-} (s^{t+3}_j)\in \{ 0,1 \} $.

Let $\Gamma : X \rightarrow \{ True,False \} $ be a function such that if $ d_{D}^{-} (x_i)=1 $, then $\Gamma( x_i)= True$ and if $ d_{D}^{-} (x_i)=2 $, then $\Gamma( x_i)= False$. We show that $\Gamma$ is a satisfying assignment
for $ \Phi $. By Fact $1$, it is enough to show that for every clause $c_j$, there exists a $s^t_j \in \{ s^1_j,s^2_j,s^3_j\}$ such that $ d_{D}^{-} (s^t_j)=2$. To the contrary suppose that $ d_{D}^{-} (s^1_j)= d_{D}^{-} (s^2_j)= d_{D}^{-} (s^3_j)=1$, so, we have $ d_{D}^{-} (s^4_j)= d_{D}^{-} (s^5_j)= d_{D}^{-} (s^6_j)=2$. Since $D$ is a proper orientation, $\{d_{D}^{-} (s^7_j) ,d_{D}^{-} (s^8_j)\}= \{ 0,1\}$, so the edges $s^7_j s^4_j$ and $ s^8_j   s^5_j $ were oriented from $s^7_j$ to $s^4_j$ and $s^8_j$ to $s^5_j$, respectively. Similarly, $\{d_{D}^{-} (s^{11}_j) ,d_{D}^{-} (s^{12}_j)\}= \{ 0,1\}$ and the edge $ s^{12}_j  s^4_j$ was oriented from $s^{12}_j$ to $s^4_j$. It means $d_{D}^{-} (s^4_j)=3$, but this is a contradiction. So $\Gamma$ is a satisfying assignment for $ \Phi $.

Next, suppose that $ \Phi $ is satisfiable with the satisfying
assignment $ \Gamma $. We present a proper orientation $D$ with maximum indegree $2$.
Let $x_i$ and $c_j$ be arbitrary variable and clause such that $x_i$ or $\neg x_i $ is used in $c_j$. For every edge $e=v_i  s^t_j$, $(v_i\in \{x_i, \neg x_i \}, s^t_j\in \{ s^1_j,s^2_j,s^3_j\})$, orient $e$ from $v_i$ to $s^t_j$.
For every $1 \leq i \leq n$, for each edge $e $ that is adjacent to $x_i^1$, orient $e$ from $x_i^1$ to the other end of $e$ and orient $x_i^2x_i^3$ from $x_i^2$ to $x_i^3$. Also if $\Gamma (x_i)=True$, orient $x_i\neg x_i$ from $x_i$ to $\neg x_i$, otherwise orient $x_i\neg x_i$ from $\neg  x_i$ to $x_i$.

In order to orient the other edges, since $\Phi$ is satisfied by $\Gamma$ and by attention to the orientations of the edges $ x_i\neg x_i$ $( 1 \leq i \leq n )$, for every clause $c_j$, at least one of its literals has indegree $1$, therefore at least one of $s^1_j$,$s^2_j$,$s^3_j$ has indegree $2$. With no loss of generality, three cases can be considered, the suitable orientations for these three cases are presented in Figure 2. This completes the proof.

}\end{aly}

\begin{figure}[h]
\begin{center}
\includegraphics[scale=.6]{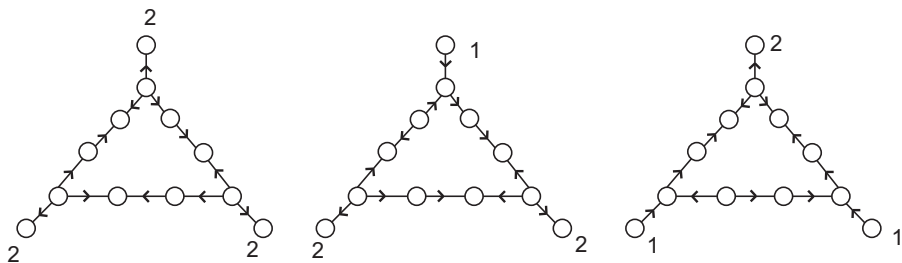}\label{a2}
\caption{The suitable orientations for three possible cases}
\end{center}
\end{figure}

Finding the optimal upper bound for $\overrightarrow{\chi}(G)$ seems to be an interesting intriguing problem. We state the following question for planar graphs.

\begin{prob}
Find the minimum number $k$ such that for every planar graph $G$, $\overrightarrow{\chi}(G)\leq k  $.
\end{prob}

\section{Acknowledgment}
\label{}

We would like to thank Professor Douglas B. West and Professor Ebad S. Mahmoodian for their valuable answers to our questions about the definition of proper orientation number.

\bibliographystyle{plain}
\bibliography{POref}

\begin{thebibliography}{10}

\bibitem{MR2145514}
L.~Addario-Berry, R.~E.~L. Aldred, K.~Dalal, and B.~A. Reed.
\newblock Vertex colouring edge partitions.
\newblock {\em J. Combin. Theory Ser. B}, 94(2):237--244, 2005.

\bibitem{MR2404230}
L.~Addario-Berry, K.~Dalal, and B.~A. Reed.
\newblock Degree constrained subgraphs.
\newblock {\em Discrete Appl. Math.}, 156(7):1168--1174, 2008.

\bibitem{DAM}
Yuichi Asahiro, Eiji Miyano, and Hirotaka Ono.
\newblock Graph classes and the complexity of the graph orientation minimizing
  the maximum weighted outdegree.
\newblock {\em Discrete Appl. Math.}, 159(7):498--508, 2011.

\bibitem{MR2311131}
Yuichi Asahiro, Eiji Miyano, Hirotaka Ono, and Kouhei Zenmyo.
\newblock Graph orientation algorithms to minimize the maximum outdegree.
\newblock {\em Internat. J. Found. Comput. Sci.}, 18(2):197--215, 2007.

\bibitem{z2}
T.~Bartnicki, J.~Grytczuk, and S.~Niwczyk.
\newblock Weight choosability of graphs.
\newblock {\em J. Graph Theory}, 60(3):242--256, 2009.

\bibitem{bondy}
J.~A. Bondy and U.~S.~R. Murty.
\newblock {\em Graph theory}.
\newblock Graduate Texts in Mathematics 244, Springer, 2008 (2nd ed.).

\bibitem{MR2895496}
Mieczys{\l}aw Borowiecki, Jaros{\l}aw Grytczuk, and Monika Pil{\'s}niak.
\newblock Coloring chip configurations on graphs and digraphs.
\newblock {\em Inform. Process. Lett.}, 112(1-2):1--4, 2012.

\bibitem{MR0012236}
R.~L. Brooks.
\newblock On colouring the nodes of a network.
\newblock {\em Proc. Cambridge Philos. Soc.}, 37:194--197, 1941.

\bibitem{MR0491326}
V.~Chv{\'a}tal and C.~Thomassen.
\newblock Distances in orientations of graphs.
\newblock {\em J. Combinatorial Theory Ser. B}, 24(1):61--75, 1978.

\bibitem{zhu1}
Ding-Zhu Du, Ker-K Ko, and J.~Wang.
\newblock {\em Introduction to Computational Complexity}.
\newblock Higher Education Press, 2002.

\bibitem{David}
Andrzej Dudek and David Wajc.
\newblock On the complexity of vertex-coloring edge-weightings.
\newblock {\em Discrete Math. Theor. Comput. Sci.}, 13(3):45--50, 2011.

\bibitem{MR2876334}
Nicole Eggemann and Steven~D. Noble.
\newblock The complexity of two graph orientation problems.
\newblock {\em Discrete Appl. Math.}, 160(4-5):513--517, 2012.

\bibitem{MR0233733}
T.~Gallai.
\newblock On directed paths and circuits.
\newblock In {\em Theory of {G}raphs ({P}roc. {C}olloq., {T}ihany, 1966)},
  pages 115--118. Academic Press, New York, 1968.

\bibitem{MR635430}
Ian Holyer.
\newblock The {NP}-completeness of edge-coloring.
\newblock {\em SIAM J. Comput.}, 10(4):718--720, 1981.

\bibitem{MR2595676}
Maciej Kalkowski, Micha{\l} Karo{\'n}ski, and Florian Pfender.
\newblock Vertex-coloring edge-weightings: towards the 1-2-3-conjecture.
\newblock {\em J. Combin. Theory Ser. B}, 100(3):347--349, 2010.

\bibitem{MR2047539}
Micha{\l} Karo{\'n}ski, Tomasz {\L}uczak, and Andrew Thomason.
\newblock Edge weights and vertex colours.
\newblock {\em J. Combin. Theory Ser. B}, 91(1):151--157, 2004.

\bibitem{MR2178357}
Sanjeev Khanna, Joseph Naor, and F.~Bruce Shepherd.
\newblock Directed network design with orientation constraints.
\newblock {\em SIAM J. Discrete Math.}, 19(1):245--257 (electronic), 2005.

\bibitem{z3}
M.~Khatirinejad, R.~Naserasr, M.~Newman, B.~Seamone, and B~Stevens.
\newblock Digraphs are 2-weight choosable.
\newblock {\em Electron. J. Combin.}, 18(1):Paper 21,4, 2011.

\bibitem{MR2296137}
{\L}ukasz Kowalik.
\newblock Approximation scheme for lowest outdegree orientation and graph
  density measures.
\newblock In {\em Algorithms and computation}, volume 4288 of {\em Lecture
  Notes in Comput. Sci.}, pages 557--566. Springer, Berlin, 2006.

\bibitem{MR689264}
Daniel Leven and Zvi Galil.
\newblock N{P}-completeness of finding the chromatic index of regular graphs.
\newblock {\em J. Algorithms}, 4(1):35--44, 1983.

\bibitem{MR2087903}
V.~Venkateswaran.
\newblock Minimizing maximum indegree.
\newblock {\em Discrete Appl. Math.}, 143(1-3):374--378, 2004.

\bibitem{MR0180505}
V.~G. Vizing.
\newblock On an estimate of the chromatic class of a {$p$}-graph.
\newblock {\em Diskret. Analiz No.}, 3:25--30, 1964.

\bibitem{MR1367739}
Douglas~B. West.
\newblock {\em Introduction to graph theory}.
\newblock Prentice Hall Inc., Upper Saddle River, NJ, 1996.

\end{thebibliography}







\end{document}